\documentclass{amsart}
\usepackage{amsmath,amssymb,amsthm,pgfplots,subcaption}
\usepackage[hidelinks]{hyperref}
\usepackage[hidelinks]{hyperref}

\newtheorem{thm}{Theorem}[section]
\newtheorem{prop}[thm]{Proposition}
\newtheorem{lem}[thm]{Lemma}

\newtheorem{cor}[thm]{Corollary}
\theoremstyle{definition}
\newtheorem{defn}[thm]{Definition}
\newtheorem{exam}[thm]{Example}
\newtheorem{rem}[thm]{Remark}
\newtheorem*{struc}{Structure of the paper}

\newcommand{\Ric}{\operatorname{Ric}}

\newcommand{\Int}{\operatorname{Int}}

\newcommand{\CB}{\mathcal{C}}
\newcommand{\real}{\mathbb{R}}
\newcommand{\nat}{\mathbb{N}}
\newcommand{\bZ}{\mathbb{Z}}

\title[Minimal Elements of the Causal Boundary]{Minimal Elements of the Causal Boundary with Applications to Spacetime Splitting}
\author{Leonardo García-Heveling}
\address{Fachbereich Mathematik, Universit\"at Hamburg, Bundesstra{\ss}e 55, 20146 Hamburg, Germany}
\email{leonardo.garcia@uni-hamburg.de}
\urladdr{https://leogarciaheveling.github.io/}
\thanks{I would like to thank  Sa\'ul Burgos, Jos\'e Luis Flores, Eric Ling, Olaf M\"uller, Arpan Saha, Miguel S\'anchez, and an anonymous referee for useful discussions and comments on the draft.}

\begin{document}

\begin{abstract}
 In 1972, Geroch, Kronheimer, and Penrose introduced what is now called the causal boundary of a spacetime. This boundary is constructed out of Terminal Indecomposable Past sets (TIPs) and their future analogues (TIFs), which are the pasts and futures of inextendible causal curves. The causal boundary is a key tool to understand the global structure of a spacetime. In this paper, we show that in a spacetime with compact Cauchy surfaces, there is always at least one minimal TIP and one minimal TIF, minimal meaning that it does not contain another TIP (resp.\ TIF) as a proper subset. We then study the implications of the minimal TIP and TIF meeting each other. This condition generalizes some of the ``no observer horizon'' conditions that have been used in the literature to obtain partial solutions of the Bartnik splitting conjecture. We also show that such a no observer horizons condition is satisfied when the spacetime has a (possibly discrete) timelike conformal symmetry, generalizing a result of Costa e Silva, Flores, and Herrera about conformal Killing vector fields.
\end{abstract}

\maketitle

\section{Introduction}

Mathematicians often like to compactify spaces by adding a point or boundary representing infinity. This is also the case in Lorentzian geometry, the framework describing spacetime in general relativity. On Lorentzian spacetimes, the notion of causal curve plays a central role, representing the physically admissible trajectories not exceeding the speed of light. Thus it is natural to consider the boundary at infinity of a spacetime to consist of the ``ideal endpoints'' to all inextendible causal curves $\gamma$. Geroch, Kronheimer and Penrose \cite{GKP} made this idea precise in the following way. On a distinguishing spacetime, there is a one-to-one correspondence between points $p$ and their past sets $I^-(p)$. Similarly, to every inextendible causal curve $\gamma$, there corresponds a past set $I^-(\gamma)$, called a TIP. The set of all TIPs constitutes the future causal boundary $\CB^+$ of the spacetime, while analogously there exists a past causal boundary $\CB^-$ made of so-called TIFs (see Section \ref{sec:prelim} for a more detailed explanation). Understanding the causal boundary is the key to many problems in Lorentzian geometry, such as the Bartnik splitting conjecture, which we discuss at the end of this introduction. The premise of this paper is to introduce and study a special kind of TIP.

\begin{defn}
 A TIP is called a \emph{minimal TIP} if it contains no other TIPs as proper subsets. Similarly, a TIF is called \emph{minimal TIF} if it contains no other TIFs as proper subsets.
\end{defn}

This definition is, to the best of our knowledge, new. The complementary notion of \emph{maximal} TIP, and, more generally, maximal IP (see Definition~\ref{def:IP}) does appear often in the literature, see e.g.\ \cite[Prop.~2.2]{Flo}. Our main result is that a spacetime with compact Cauchy surfaces always contains a minimal TIP and a minimal TIF.

\begin{thm} \label{thm:mTIP}
 Let $(M,g)$ be a spacetime containing a compact Cauchy surface. Then every TIP contains a minimal TIP, and every TIF contains a minimal TIF. In particular, there exists at least one minimal TIP and one minimal TIF in the spacetime.
\end{thm}

Notice that Minkowski spacetime constitutes a counterexample when the assumption of compactness of the Cauchy surface is dropped. As a first application, we show that when a minimal TIP and a minimal TIF meet, then the spacetime contains a timelike line. That is, an inextendible maximizing timelike geodesic.

\begin{defn}
 We say that a spacetime $(M,g)$ satisfies the \emph{min-min condition} if it contains a minimal TIP and a minimal TIF whose intersection is non-empty.
\end{defn}

\begin{thm} \label{thm:line}
 Let $(M,g)$ be globally hyperbolic with compact Cauchy surfaces. If $(M,g)$ satisfies the min-min condition, then $(M,g)$ contains a timelike line.
\end{thm}

The min-min condition is a generalization of the no observer horizon condition (NOH). Physically speaking, the Universe contains observer horizons if there are observers that never recieve signals from certain spacetime regions, the horizon being the boundary of the visible region. Signals from beyond the horizon never arrive, either because the Universe is expanging faster than the speed of light, or because the observer accelerates away and approaches the speed of light. Mathematically, the NOH states that $I^\pm(\gamma) = M$ for every inextendible causal curve; equivalently, the past and future causal boundary is a single point, $\CB^\pm = \{M\}$. In particular, it then follows that $M$ is a minimal TIP and a minimal TIF.

Costa e Silva, Flores and Herrera \cite{CSFH1} showed that every spacetime with compact Cauchy surfaces that admits a complete timelike conformal Killing vector field $X$ satisfies the NOH (see also the earlier, related work of Galloway and Vega \cite{GaVe3}). We generalize this result to the case where the spacetime has a timelike conformal symmetry (which include, among others, the flows of complete conformal Killing vector fields). The proof in \cite{CSFH1} is based on a special decomposition of the metric obtained earlier \cite{CFS,JaSa} using a temporal function with gradient along the conformal Killing field. One can then, in some sense, explicitly compute the causal boundary. Our proof is completely different, and uses the notion of minimal TIP in a crucial way. Note also that Galloway obtained a similar result \cite[Thm.~4.1]{GalSym} to our Theorem \ref{thm:conf} below, but using a different definition of discrete timelike symmetry (see also references in \cite[Sec.~4]{GalSym} for other treatments of discrete timelike symmetries).

\begin{defn} \label{def:conf}
 Let $(M,g)$ be a spacetime. A \emph{timelike conformal transformation} is a diffeomorphism $\phi : M \to M$ such that $\phi^* g = \Omega g$ for some positive function $\Omega \in \mathcal{C}^\infty(M)$ and $x \ll \phi(x)$ for every $x \in M$.
\end{defn}

\begin{thm} \label{thm:conf}
 Let $(M,g)$ be a connected chronological spacetime admitting a timelike conformal transformation. Then $M$ is a TIP and a TIF of $(M,g)$, i.e. $M \in \CB^\pm$. If, additionally, $(M,g)$ is globally hyperbolic with compact Cauchy surfaces, then the future and past causal boundaries consist of a single point, $\CB^\pm = \{M\}$, and $(M,g)$ contains a timelike line.
\end{thm}

The existence of a timelike line is well-known as one of the assumptions of the Lorentzian splitting theorem \cite{Esch} (see also \cite{GalSplit,New}), leading to the following corollary.

\begin{cor} \label{cor:Bartnik}
 Let $(M,g)$ be a connected spacetime satisfying:
 \begin{enumerate}
  \item global hyperbolicity with compact Cauchy surfaces,
  \item the timelike convergence condition: $\Ric(X,X) \geq 0$ for all timelike $X$,
  \item timelike geodesic completeness,
  \item the min-min condition or existence of a timelike conformal transformation.
 \end{enumerate}
Then $(M,g)$ is isometric to a product $(\real \times S, -dt^2+h)$, with $h$ a Riemannian metric on $S$ independent of $t$.
\end{cor}

This corollary constitutes a special case of the Bartnik splitting conjecture \cite[Conj.~2]{Bar}, which is the same statement but assuming only (i)-(iii). Bartnik's conjecture can be understood as a rigid generalization of the Hawking--Penrose singularity theorem \cite{GalRev}. The latter states that a spacetime with compact Cauchy surfaces that satisfies the timelike convergence condition and the \emph{generic} condition must be causally geodesically incomplete. Bartnik's splitting conjecture thus asserts that when removing the generic assumption, the only exception to incompleteness occurs in the ``rigid'' case of a product spacetime. The conjecture can also be given a similar interpretation as a rigid Hawking singularity theorem, where one has dropped the assumption on the mean curvature of the Cauchy surface. In fact, the conjecture is closely related to the question of when a spacetime contains a constant mean curvature (CMC) Cauchy surface. This is because under assumptions (i)-(iii), a CMC with mean curvature $0$ leads to splitting of the spacetime \cite[Cor.~2]{Bar}, while non-vanishing constant mean curvature would violate timelike completeness by Hawking's singularity theorem.

The Bartnik splitting conjecture has so far been proven under various extra assumptions, usually by employing them to verify the no observer horizons condition (NOH), which then reduces the problem to standard arguments \cite{ASB,Bar,CSFH1,CSFH2,EhGa,GalWarsaw,GaVe1,GaVe2,GaVe3,ShBa,Tip}. In fact, proving splitting is equivalent to proving the NOH, since a rigid product always satisfies it. Ehrlich and Galloway \cite{EhGa}, however, gave an example of a NOH-violating spacetime that satisfies (i) and (iii). Thus the difficulty remains of establishing whether the curvature condition (ii) forces the NOH to be satisfied. It is known that non-positive timelike sectional curvature (which is stronger than (ii)) is sufficient \cite{EhGa, GaVe2}.

\begin{struc}
 In Section \ref{sec:prelim}, we give some necessary background on the causal boundary. Then we prove Theorem \ref{thm:mTIP} in Section \ref{sec:mTIP}. In Section \ref{sec:disc}, we prove Theorem \ref{thm:line} and discuss it in the context of known splitting results. Finally, Section \ref{sec:conf} contains the proof of Theorem \ref{thm:conf} and three examples: two of discrete timelike conformal transformation, and one of a NOH-violating spacetime that admits a complete causal conformal Killing vector field.
\end{struc}

\section{Preliminaries about the causal boundary} \label{sec:prelim}

Throughout the manuscript, $(M,g)$ denotes a smooth spacetime. We assume a certain familiarity with standard causal theory of spacetimes (see \cite{MinRev} for a review), but not with the causal boundary. Recall that the past $I^-(A)$ of a set $A \subseteq M$ is defined as
\[
 I^-(A) := \{ p \in M \mid \exists q \in A \text{ such that } p \in I^-(q) \}.
\]
Analogously we can define the future of a set, and also the rest of this section can be formulated in a time reversed fashion. We omit these time reversed statements for brevity.
\begin{defn}[\cite{GKP}] \label{def:IP}
 A subset $U \subseteq M$ is called:
 \begin{enumerate}
  \item A \emph{past set} if $I^-(U) = U$.
  \item An \emph{indecomposable past set (IP)} if it is a past set that cannot be written as the union of two proper past subsets.
  \item A \emph{terminal indecomposable past set (TIP)} if it is an indecomposable past set that is \underline{not} of the form $I^-(p)$ for any $p \in M$.
 \end{enumerate}
\end{defn}

Notice that the past $I^-(p)$ of a point $p$ can also be thought of as the past $I^-(\gamma)$ of any future directed causal curve $\gamma$ ending at $p$. Similarly, the following lemma tells us that TIPs can be thought of as the ``endpoints'' of inextendible causal curves. Note that a causal curve $\gamma : (a,b) \to M$ is future-inextendible if and only if $\lim_{s \to b} \gamma(s)$ does not exist.

\begin{lem}[{\cite[Thms.~2.1 \& 2.3]{GKP}}] \label{lem:TIP}
 Let $(M,g)$ be a distinguishing spacetime. For any $U \subseteq M$, the following are equivalent:
 \begin{enumerate}
  \item $U$ is a TIP.
  \item $U = I^-(\gamma)$ for a future-inextendible timelike curve $\gamma$.
  \item $U = I^-(\gamma)$ for a future-inextendible causal curve $\gamma$.
 \end{enumerate}
\end{lem}

Motivated by this result, we call the set of all TIPs the \emph{future causal boundary} $\CB^+$ of $(M,g)$. Analogously, there is a past causal boundary $\CB^-$ made of TIFs (terminal indecomposable future-sets, which can be identified as $I^+(\gamma)$ for past-inextendible causal curves $\gamma$). In general, one might want to identify some TIPs with some TIFs as representing the same ``ideal point''. For example, if $(M,g)$ is a spacetime and $p \in M$ a point, the set $I^-(p)$ (resp.\ $I^+(p)$) constitutes a TIP (resp.\ TIF) in the spacetime $(M \setminus \{p \}, g\vert_{M \setminus \{p\}})$, but one should think of $I^+(p)$ and $I^-(p)$ as both representing the same missing point $p$. How to perform these identifications correctly on an arbitrary spacetime (and how to equip the ensuing causal boundary with a well-behaved topology) is a non-trivial question that has received considerable attention in the literature, see e.g.\ \cite{BFH,CSFH3,Flo,FHS,Har,MaRo} and references therein. In the present paper, however, we restrict to globally hyperbolic spacetimes, where no identifications between $\CB^+$ and $\CB^-$ are required \cite[Thm.~3.29]{FHS}, and moreover, we will not need a topology at all. We do make use of the following natural relation on the causal boundary.

\begin{defn}[\cite{GKP}] \label{def:orderCB}
 Let $U,V$ be TIPs. We say that a $U$ is \emph{in the past of} $V$ if $U \subseteq V$. Moreover, we say that the future causal boundary is \emph{spacelike} if $U \subseteq V \implies U = V$.
\end{defn}

Again, the definition is motivated by an analogy to the case of two points, where $p$ being in the past of $q$ implies $I^-(p) \subseteq I^-(q)$, the converse being true on causally simple spacetimes. Recall that $I^-(p)$ is open for every $p \in M$. It follows that every past set $U$ is open, as it is the union of open sets. The closure $\overline{U}$ enjoys the following property, which will be useful in multiple parts of the paper.

\begin{lem} \label{lem:limitpast}
 If $U \subseteq M$ is a past set and $p \in \overline{U}$, then $I^-(p) \subseteq U$. In particular, if $(\gamma_n)_n$ is a sequence of causal curves in $U$ converging to a limit curve $\gamma$, then $I^-(\gamma) \subseteq U$.
\end{lem}

\begin{proof}
 Let $p_n$ be a sequence in $U$ with $p_n \to p$, and let $x \in I^-(p)$ be arbitrary. Then $p \in I^+(x)$, and since $I^+(x)$ is open, we have $p_n \in I^+(x)$ for all large enough $n$. But then $x \in I^-(p_n)$, which implies $x \in I^-(U) = U$ because $p_n \in U$. The statement about limit curves follows directly from the statement about limit points.
\end{proof}

\section{Existence of minimal TIPs and TIFs} \label{sec:mTIP}

We prove Theorem \ref{thm:mTIP} for TIPs, the version for TIFs being analogous. Since $(M,g)$ contains a Cauchy surface, it is globally hyperbolic and admits a Cauchy temporal function $\tau : M \to \real$ whose level sets define a foliation of $M$ by (in our case compact) Cauchy surfaces \cite{Ger}. Furthermore, we equip $M$ with a finite Borel measure of full support $\mu$ (see e.g.\ \cite[f.n.~24]{Ger} for an existence argument). Then $\mu$ distinguishes past sets (in particular, TIPs) in the following sense.

\begin{lem} \label{lem:mu}
 If $U$ and $V$ are past sets and $U \subseteq V$, then $\mu(U) = \mu(V)$ only if $U=V$.
\end{lem}

\begin{proof}
 Assume that $\mu(U) = \mu(V)$ but $U \subsetneq V$, so there exists a point $p \in V \setminus U$. Since $V$ is a past-set, it must be open. Therefore $W := V \cap I^+(p)$ is open and non-empty, and moreover $W \subseteq V \setminus U$, since otherwise we would have $p \in U$. But then
 \begin{equation*}
  \mu(V) = \mu(U) + \mu(V \setminus U) \geq \mu(U) + \mu(W) > \mu(U),
 \end{equation*}
 in contradiction to our initial assumption.
\end{proof}

Lemma \ref{lem:mu} shows that the function $\CB^+ \to \real,\ U \mapsto \mu(U)$ is strictly monotonous with respect to the inclusion on $\CB^+$. Our strategy to prove the existence of a minimal TIP (Theorem \ref{thm:mTIP}) is to show that this function attains a minimum. One way to show this is using the Beem topology on the causal boundary and its lower semi-compactness \cite[Thm.~12]{Mue}. However, we instead opt for a direct construction of the minimizer.

\begin{proof}[Proof of Theorem \ref{thm:mTIP}]
 We recursively construct a nested sequence of TIPs as follows. Let $U_0$ be any TIP in our spacetime. Given an element of the sequence $U_n$, let
 \begin{align*}
 \mu_n := \mu(U_n), \quad \nu_n := \inf\{\mu(V) \mid V \subseteq U_n \text{ is a TIP} \}.
 \end{align*}
 Note that $\nu_n \leq \mu_n$. Choose the next sequence element $U_{n+1}$ to be such that
 \begin{equation*}
  \mu_{n+1} \leq \frac{\mu_n + \nu_n}{2}.
 \end{equation*}
 We can do this because $\mu_n \geq \nu_n$ and hence the right hand side is greater than or equal to $\nu_n$, with equality if and only if $\mu_n = \nu_n$, in which case $U_n$ is the sought-after minimal TIP by Lemma \ref{lem:mu}. It remains to treat the case when we obtain an infinite sequence with $\nu_n < \mu_n$ for all $n \in \nat$. Then also $\mu_n > \mu_{n+1}$ and $\nu_n \leq \nu_{n+1}$ (because at each successive step, the infimum is over a smaller set of subsets). We claim that
 \begin{equation*}
  U := \Int \left( \bigcap_{n \in \nat} U_n \right)
 \end{equation*}
 is the desired minimal TIP (here $\Int$ denotes the interior).

First we show that $U$ is a past-set (if it is non-empty). Let $x \in U$ be arbitrary, and $y \in I^-(x)$. Then since $x \in U_n$ for all $n \in \nat$, by transitivity also $y \in I^-(U_n) = U_n$ for all $n \in \nat$. By openness of $I^-(x)$, the same argument applies to all points in a neighborhood of $y$, hence $y \in U$. This proves that $I^-(U) \subseteq U$. On the other hand, $U \subseteq I^-(U)$ is true for every open set.

We proceed to prove that $U$ is non-empty and intersects every $\tau$-level set. Let $\Sigma_0, \Sigma_1$ be the $\tau$-level sets for two arbitrary values $\tau_0 < \tau_1$. Every $U_n$, being a TIP, must intersect the Cauchy surface $\Sigma_1$ by Lemma \ref{lem:TIP} (since we can choose the representing curve to be inextendible). Thus we may choose a sequence $p_n \in \Sigma_1 \cap U_n$, and by compactness of $\Sigma_1$, a subsequence, denoted again by $(p_n)_n$, has a limit $p \in \Sigma_1$. Because $\Sigma_0, \Sigma_1$ are Cauchy and $\tau_0 < \tau_1 = \tau(p)$, it holds that $I^-(p) \cap \Sigma_0 \neq \emptyset$. Choose any $q \in I^-(p) \cap \Sigma_0$. Then $p \in I^+(q)$, and by openness of the chronological future, $p_n \in I^+(q)$ for all large enough $n \in \nat$. Therefore $q \in I^-(p_n)$ for all large enough $n \in \nat$, implying that $q \in U_n$ (because $p_n \in U_n$ and $U_n$ is a past set). Since the $U_n$ are nested, it follows that $q \in \bigcap_n U_n$, and since the same argument applies to every point in a small enough neighborhood of $q$, it follows that $q \in U$. Since we can do this for every $\tau_0 \in \real$ (and suitable $\tau_1$, for instance $\tau_1 = \tau_0 + 1$), it follows that $U$ intersects every $\tau$-level set.

Knowing that $U$ is a past set that intersects every $\tau$-level set, we can construct a sequence $\gamma_k$ of past-directed causal curves in $U$, starting on $\Sigma_k$ and ending on $\Sigma_0$. Since $\Sigma_0$ is compact, a subsequence accumulates at some point, and by the limit curve theorem \cite[Thm.~3.1(1)]{MinLim}, there is a future inextendible limit curve $\gamma$ contained in $\overline{U}$. Then $I^-(\gamma) \subseteq U$ by Lemma \ref{lem:limitpast}.

Finally, we show that $I^-(\gamma)=U$ is the minimal TIP we are looking for. Observe that
\begin{equation*}
 \mu(U) = \lim_{n \to \infty} \mu_n = \lim_{n \to \infty} \nu_n.
\end{equation*}
Thus any TIP $V$ contained in $U$ must have $\mu(V) = \mu(U)$, since otherwise there would be an $n$ such that $\mu(V) < \nu_n$, but at the same time $V \subseteq U \subseteq U_n$, in contradiction to the definition of $\nu_n$. Since, by Lemma \ref{lem:mu}, $\mu(V) = \mu(U) \implies V=U$, we have shown that $I^-(\gamma)=U$ is a minimal TIP.
\end{proof}

\begin{rem}[Alternative proof]
 One can also show the existence of a minimal TIP by applying Zorn's lemma to the future causal boundary $\CB^+$, partially ordered by reverse inclusion $\supseteq$. To do this, one needs to show that every chain has an upper bound, i.e., that for every totally ordered subset $\{U_\alpha\}_{\alpha \in I}$ of $(\CB^+,\supseteq)$ (where $I$ can be uncountable), there is a TIP $U$ with $U_\alpha \supseteq U$ for all $\alpha \in I$. This can be done as follows. For every $U_\alpha$, there is an inextendible causal curve $\gamma_\alpha$ such that $U_\alpha = I^-(\gamma_\alpha)$. By using a finite measure as in the proof above, we can extract a countable sequence $V_j = U_{\alpha_j}$, such that for every $\alpha \in I$ there is a $j \in \nat$ such that $U_\alpha \supseteq V_j$. Then, we can apply the limit curve theorem to the sequence $(\gamma_{\alpha_j})_j$ (because all these curves intersect a compact Cauchy surface), and obtain a limit curve $\gamma$. Finally, $U = I^-(\gamma)$ is the desired upper bound of $\{U_\alpha\}_{\alpha \in I}$.
\end{rem}

\section{The min-min condition and spacetime splitting} \label{sec:disc}

We first give a direct proof of Theorem \ref{thm:line}, that is, that the min-min condition implies the existence of a timelike line. Below, we also give an alternative proof using a result of Galloway \cite[Thm.~4.1]{GalWarsaw}. In the rest of the section, we prove implications between min-min and other conditions from the literature, and provide examples.

\begin{proof}[Proof of Theorem \ref{thm:line}]
 Let $\gamma_1 : \real \to M$ and $\gamma_2 : \real \to M$ be future-directed timelike curves such that $I^-(\gamma_1)$ is a minimal TIP, $I^+(\gamma_2)$ is a minimal TIF, and $I^-(\gamma_1) \cap I^+(\gamma_2) \neq \emptyset$. The intersection being non-empty means that we may assume, without loss of generality, that $\gamma_1(1) \in I^+(\gamma_2(-1))$ and $\tau(\gamma_2(-1)) < 0 < \tau(\gamma_1(1))$. Therefore $\gamma_2(-n)$ is in the past of $\gamma_1(n)$ for every $n \in \nat$.

 Because $(M,g)$ is globally hyperbolic, for every $n$ there exists a length maximizing timelike geodesic $\sigma_n$ from $\gamma_2(-n)$ to $\gamma_1(n)$. For each $n$, $\sigma_n$ must intersect the level set $\Sigma_0 := \tau^{-1}(0)$, which is compact. Therefore the curves $\sigma_n$ accumulate at some point of $\Sigma_0$, and by the limit curve theorem and the lower semicontinuity of the length functional \cite[Thm.~3.1(1)]{MinLim}, the sequence $\sigma_n$ converges to a limit curve $\sigma$ which is an inextendible, length maximizing causal geodesic (i.e.\ a line).

 It remains to show that $\sigma$ is timelike. Since each $\sigma_n$ lies in $I^-(\gamma_1)$, it follows from Lemma \ref{lem:limitpast} that $I^-(\sigma) \subseteq {I^-(\gamma_1)}$. But because $I^-(\gamma_1)$ is a minimal TIP, we must in fact have $I^-(\sigma) = I^-(\gamma_1)$. Similarly, $I^+(\sigma) = I^+(\gamma_2)$. It follows that $I^-(\sigma) \cap I^+(\sigma) = I^-(\gamma_1) \cap I^-(\gamma_2) \neq \emptyset$. Hence $\sigma$ is not achronal, so it cannot be a null line, and must therefore be timelike.
\end{proof}


In \cite[Thm.~4.1]{GalWarsaw}, Galloway proves the Bartnik splitting conjecture under an additional hypothesis known as the ray-to-ray condition (defined below). Recall that, given a Cauchy surface $S$, a future $S$-ray is a future-inextendible geodesic starting on $S$ that maximizes the Lorentzian distance to $S$. Similarly, one defines past $S$-rays. One can show using a limit curve argument that compact Cauchy surface $S$ always admits a future $S$-ray and a past $S$-ray, but the starting points of these need not be the same. If they do coincide, the concatenation of the two rays is a line. A weaker assumption in order to obatin a line is the following.

\begin{defn}
 A spacetime $(M,g)$ is said to satisfy the \emph{ray-to-ray} condition if it contains a Cauchy hypersurface $S$, a future $S$-ray $\gamma$, and a past $S$-ray $\sigma$ such that $I^-(\gamma) \cap I^+(\sigma) \neq \emptyset$.
\end{defn}

\begin{prop} \label{prop:ray}
 Let $S$ be a compact Cauchy surface in a spacetime $(M,g)$, and let $U$ be a minimal TIP. Then $U = I^-(\gamma)$ for some future $S$-ray $\gamma$. Similarly, every TIF can be represented as the future of a past $S$-ray. In particular, the min-min condition implies the ray-to-ray condition.
\end{prop}

\begin{proof}
 We prove the statement for TIPs, writing $U = I^-(\sigma)$ for some timelike curve $\sigma \colon \real \to M$. We may assume, without loss of generality, that $\sigma(0) \in I^+(S)$. Therefore there is a sequence of timelike geodesics $\gamma_n$ that realize the distance between $S$ and $\sigma(n)$. By compactness of $S$ and the limit curve theorem, there exists a limit curve $\gamma$ that is an $S$-ray (cf.~\cite[Thm.~2.65]{MinRev}). By construction, $\gamma$ is contained in $\overline{U}$, so by Lemma \ref{lem:limitpast}, $I^-(\gamma) \subseteq U$. Since $U$ is minimal, in fact $I^-(\gamma) = U$.
\end{proof}

Combining Proposition \ref{prop:ray} with Galloway's result \cite[Thm.~4.1]{GalWarsaw} yields an alternative proof of Corollary \ref{cor:Bartnik}. While the ray-to-ray condition is weaker than the min-min condition, the min-min condition generalizes other assumptions that had previously been made to prove special cases of Bartnik's conjecture. Note also that the min-min condition, unlike the ray-to-ray condition, is a pure causality condition (i.e.\ conformally invariant).

\begin{prop} \label{prop:impliesminmin}
 Let $(M,g)$ be a globally hyperbolic spacetime with compact Cauchy surfaces. Any of the following conditions implies the min-min condition:
 \begin{enumerate}
  \item There is a point $p \in M$ such that $M \setminus \left(I^+(p) \cup I^-(p) \right)$ is compact (Bartnik \cite{Bar}).
  \item The future causal boundary is a single point (Tipler \cite{Tip}).
  \item The future causal boundary is spacelike (Galloway and Vega \cite[Thm.~5.10]{GaVe2}).
  \item There is a Cauchy surface $S \subset M$ such that $S \subset I^-(\gamma)$ for every future-inextendible causal curve $\gamma$.
 \end{enumerate}
\end{prop}

The references given correspond to papers where the relevant condition was introduced or studied. Of course, the time-reversed versions of the above conditions also imply the min-min condition.

\begin{proof}
 By Theorem \ref{thm:mTIP}, there is a minimal TIP $U = I^-(\gamma)$ and a minimal TIF $V = I^+(\sigma)$, for $\gamma$ and $\sigma$ two inextendible timelike curves.

 \textit{(i)} $\implies$ \textit{min-min}. Since $(M,g)$ is globally hyperbolic, $\gamma$ and $\sigma$ cannot be entirely contained in the compact set $M \setminus \left(I^+(p) \cup I^-(p) \right)$. Therefore $\gamma$ must eventually enter the future $I^+(p)$, since if it stayed forever in $I^-(p)$, compactness of the causal diamonds would be violated. Similarly, $\sigma$ must enter $I^-(p)$ (towards the past). But then $p \in U \cap V$, so $U \cap V \neq \emptyset$.

 \textit{(ii)} $\implies$ \textit{(iii)} $\implies$ \textit{min-min}. Since $\sigma$ is inextendible, $I^-(\sigma)$ is a TIP, which must be minimal, because (iii) implies that all TIPs are minimal. Since $I^+(\sigma)$ is a minimal TIF by assumption, and clearly $I^+(\sigma) \cap I^-(\sigma) \neq \emptyset$ by timelikeness of $\sigma$, we are done.

 \textit{(iv)} $\implies$ \textit{min-min}. On the one hand, $S \subset U$ by assumption. On the other hand, $V \cap S \neq \emptyset$ since $\sigma$ and hence $V$ must intersect every Cauchy surface. Therefore $U \cap V \neq \emptyset$, as claimed.
\end{proof}

Note that any of the above conditions (including min-min and ray-to-ray) is satisfied on a Lorentzian product. Thus, if Bartnik's conjecture is true, then all of these conditions are, in fact, equivalent (within the class of spacetimes satisfying the conjecture's assumptions). On the other hand, Example \ref{ex:sharp} below shows that, without assuming the timelike convergence condition, the min-min condition is strictly weaker than (i)-(iv).

\begin{exam} \label{ex:sharp}
 Consider the region in $(1+1)$-dimensional Minkowski spacetime given by $-3 \leq x \leq 3$ and $\max\{-1-\vert x \vert,-2\} < t < \min\{ \vert x \vert, 1 \}$, and identify $(-3,t) \sim (3,t)$ for all $-2 < t < 1$. See Figure \ref{fig:sharp} for a visual representation. The future causal boundary corresponds one-to-one to the set $\{ t = \min\{ \vert x \vert, 1 \} \}$, while the past causal boundary corresponds to $\{ t = \max\{-1-\vert x \vert,-2\} \}$. It is then easy to see that $U := \{t < -\vert x \vert\}$ is a minimal TIP, while $V := \{t > \vert x \vert-1\}$ is a minimal TIF, and $U \cap V = \{\vert x \vert-1 < t < -\vert x \vert \} \neq \emptyset$, so the min-min condition is satisfied. Note that there are more (intersecting) minimal TIPs and TIFs, since every TIP and TIF representing a point where the causal boundary is spacelike will be minimal.

 On the other hand, the conditions (i)-(iv) from Proposition \ref{prop:impliesminmin} are not satisfied: If (i) were satisfied, then the point $p$ with $M \setminus \left(I^+(p) \cup I^-(p) \right)$ compact would have to lie in $U \cap V$, because otherwise $U \cap V \subseteq M \setminus \left(I^+(p) \cup I^-(p) \right)$, which is not possible for a compact set. But if $p \in U \cap V$, then  $M \setminus \left(U \cup V\right) \subset M \setminus \left(I^+(p) \cup I^-(p) \right)$, which also contradicts compactness. Clearly, (ii) and (iii) are not satisfied either, since the causal boundary contains a spacelike and a lightlike part. Finally, (iv) is no satisfied, because one can find an inextendible causal curve contained in $U$, and then (iv) would imply that the Cauchy surface $S$ is contained in $U$. But one can also find inextendible causal curves that never meet $U$, a contradiction.
\end{exam}

 \begin{figure}
  \centering
  \begin{tikzpicture}[scale=1.5]
   \draw[thick] (-3,1) -- (-1,1) -- (0,0) -- (1,1) -- (3,1) -- (3,-2) -- (1,-2) -- (0,-1) -- (-1,-2) -- (-3,-2) -- cycle;
   \draw[thick,dashed] (-2,1) -- (0,-1) -- (2,1);
   \draw[thick,dashed] (-2,-2) -- (0,0) -- (2,-2);
   \node at (0,-0.5) {\large $U \cap V$};
   \node at (1,-1.5) {\large $U$};
   \node at (-1,-1.5) {\large $U$};
   \node at (1,0.5) {\large $V$};
   \node at (-1,0.5) {\large $V$};
   \node at (0,-2.1) {\large $\CB^-$};
   \node at (0,1.1) {\large $\CB^+$};
   \node[rotate=90] at (3.2,-0.5) {identify};
   \node[rotate=90] at (-3.2,-0.5) {identify};
   \draw[thick,->] (-3,-2) -- (-3,-0.5);
   \draw[thick,->] (3,-2) -- (3,-0.5);
  \end{tikzpicture}
 \caption{The spacetime of Example \ref{ex:sharp}, with its intersecting minimal TIP $U$ and minimal TIF $V$.}
 \label{fig:sharp}

 \end{figure}
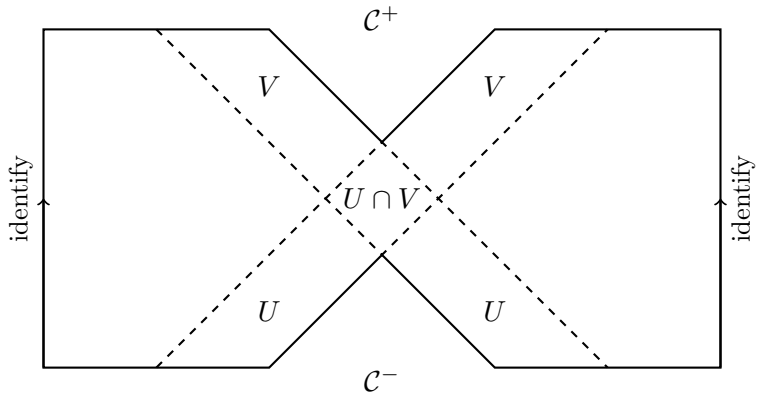

We end the paper with an example, borrowed from Ehrlich and Galloway \cite{EhGa}, of a spacetime $(M,g)$ with compact Cauchy surfaces which does not satisfy the min-min condition. Ehrlich and Galloway show that one can choose a conformal factor $\Omega$ such that $(M, \Omega g)$ is geodesically complete but does not contain any timelike lines (and hence does not split). However, $(M,\Omega g)$ does not satisfy the timelike convergence condition, either. This could not be any other way, since Bartnik's splitting conjecture is known to hold in $(1+1)$-dimensions (this follows from the result for sectional curvature bounds in \cite{EhGa}).

\begin{exam}[\cite{EhGa}] \label{ex:notminmin}
 Consider the region in $(1+1)$-dimensional Minkowski spacetime given by $-1 \leq x \leq 1$ and $-1-\vert x \vert < t < 1 - \vert x \vert$, and identify $(-1,t) \sim (1,t)$ for all $-2 < t < 0$. See Figure \ref{fig:notminmin} for a visual representation. The future causal boundary corresponds one-to-one to the set $\{ t = 1 - \vert x \vert \}$, while the past causal boundary corresponds to $\{ t = -1 - \vert x \vert \}$. It is then easy to see that $U := \{t < \vert x \vert-1\}$ forms the only minimal TIP in the spacetime, while $V := \{t > \vert x \vert-1\}$ forms the only minimal TIF, and $U \cap V = \emptyset$, so the min-min condition is not satisfied.
\end{exam}

 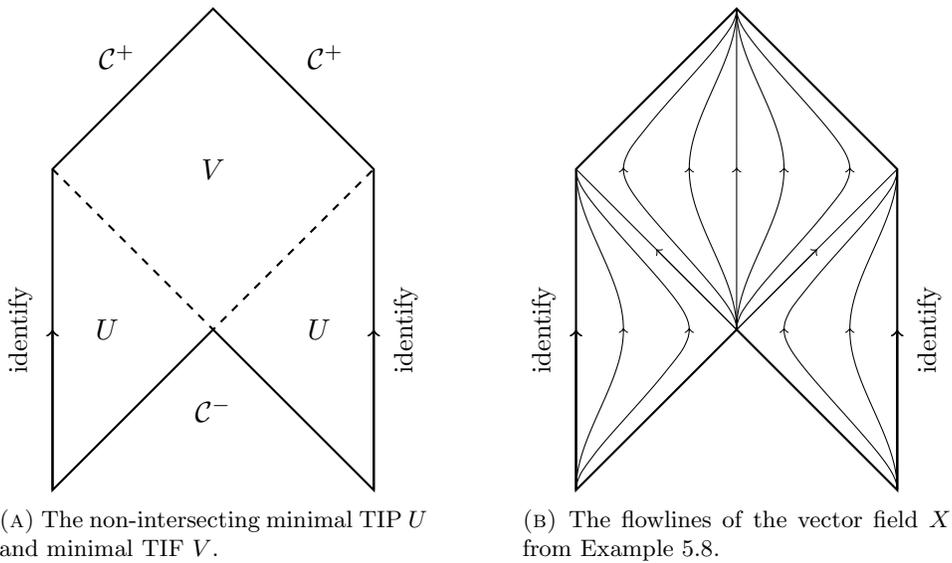
\begin{figure}
 \begin{subfigure}{0.45\textwidth}
  \centering
  \begin{tikzpicture}
  \begin{axis}[xmin=-2, xmax=2, ymin=-2, ymax=2, hide axis,axis equal image,clip=false,scale=1.5]
   \draw[thick] (axis cs:-1,0) -- (axis cs:0,1) -- (axis cs:1,0) -- (axis cs:1,-2) -- (axis cs:0,-1) -- (axis cs:-1,-2) -- cycle;
   \draw[thick,dashed] (axis cs:-1,0) -- (axis cs:0,-1) -- (axis cs:1,0);
   \node at (axis cs:0,0) {\large $V$};
   \node at (axis cs:0.66,-1) {\large $U$};
   \node at (axis cs:-0.66,-1) {\large $U$};
   \node at (axis cs:0,-1.5) {\large $\CB^-$};
   \node at (axis cs:-0.6,0.7) {\large $\CB^+$};
   \node at (axis cs:0.7,0.7) {\large $\CB^+$};
   \node[rotate=90] at (axis cs:1.2,-1) {identify};
   \node[rotate=90] at (axis cs:-1.2,-1) {identify};
   \draw[thick,->] (axis cs:-1,-2) -- (axis cs:-1,-1);
   \draw[thick,->] (axis cs:1,-2) -- (axis cs:1,-1);
   \end{axis}
  \end{tikzpicture}
 \caption{The non-intersecting minimal TIP $U$ and minimal TIF $V$.}
 \label{fig:notminmin}
 \end{subfigure} \hfill
 \begin{subfigure}{0.45\textwidth}
    \centering
  \begin{tikzpicture}
  \begin{axis}[xmin=-2, xmax=2, ymin=-2, ymax=2, hide axis,axis equal image,clip=false,scale=1.5]
   \draw[thick] (axis cs: -1,0) -- (axis cs: 0,1) -- (axis cs: 1,0) -- (axis cs: 1,-2) -- (axis cs: 0,-1) -- (axis cs: -1,-2) -- cycle;
   \draw (axis cs: -1,0) -- (axis cs: 0,-1) -- (axis cs: 1,0);
   \draw[->] (axis cs: 0,-1) -- (axis cs: 0.5,-0.5);
   \draw[->] (axis cs: 0,-1) -- (axis cs: -0.5,-0.5);
   \node[rotate=90] at (axis cs: 1.2,-1) {identify};
   \node[rotate=90] at (axis cs: -1.2,-1) {identify};
   \draw[thick,->] (axis cs: -1,-2) -- (axis cs: -1,-1);
   \draw[thick,->] (axis cs: 1,-2) -- (axis cs: 1,-1);
   \foreach \k in {-2,...,2}
   {
   \addplot[domain=-1.57:1.57,samples=50,smooth,variable=\x] ({(rad(atan(\k^2+tan(deg(x))))-x)/pi},{(rad(atan(\k^2+tan(deg(x))))+x)/pi});
   \addplot[->,domain=-0.01:0.01,samples=2,variable=\x] ({rad(atan(0.5*\k^2)-atan(-0.5*\k^2))/pi},{x});
   }
   \foreach \k in {-2,...,-1}
   {
   \addplot[domain=-1.57:1.57,samples=50,smooth,variable=\x] ({(rad(atan(\k^2+tan(deg(x))))-x)/pi+1},{(rad(atan(\k^2+tan(deg(x))))+x)/pi-1});
   \addplot[->,domain=-0.01:0.01,samples=2,variable=\x] ({rad(atan(0.5*\k^2)-atan(-0.5*\k^2))/pi+1},{x-1});
   }
   \foreach \k in {1,...,2}
   {
   \addplot[domain=-1.57:1.57,samples=50,smooth,variable=\x] ({(rad(atan(\k^2+tan(deg(x))))-x)/pi-1},{(rad(atan(\k^2+tan(deg(x))))+x)/pi-1});
   \addplot[->,domain=-0.01:0.01,samples=2,variable=\x] ({rad(atan(0.5*\k^2)-atan(-0.5*\k^2))/pi-1},{x-1});
   }
  \end{axis}
  \end{tikzpicture}
 \caption{The flowlines of the vector field $X$ from Example \ref{ex:flow}.}
 \label{fig:flow}
 \end{subfigure}
 \caption{The spacetime of Examples \ref{ex:notminmin} and \ref{ex:flow}.}
 \label{fig:ex}
 \end{figure}

\section{Timelike conformal symmetries} \label{sec:conf}

In this section, we prove Theorem \ref{thm:conf} (subdivided into lemmas), and then give some examples. Recall from Definition \ref{def:conf} that a conformal diffeomorphism $\phi$ with the property that $x \ll \phi(x)$ for every $x \in M$ is called timelike. A timelike conformal transformation induces a $\bZ$-action whose orbits are timelike chains, just as a complete timelike conformal Killing vector field induces an $\real$-action whose orbits are timelike curves. We start by proving a series of lemmas for timelike conformal transformations.

\begin{lem} \label{lem:timeor}
 Let $(M,g)$ be a chronological spacetime with a timelike conformal transformation $\phi : M \to M$. Then $\phi$ preserves the time orientation and the causal and chronological relation.
\end{lem}

\begin{proof}
 A conformal transformation always preserves causal and timelike vectors and curves. Thus $\phi$ preserves the causal and chronological relation if and onlf if it preserves the time orientation. Let $X$ be the timelike vector field that determines the time orientation of $(M,g)$. Note that there are only two possible time orientations on a spacetime (the other one given by $-X$). Since $d\phi(X)$ is a timelike vector field on $M$, it determines a time orientation, and we need to prove that it conicides with the one determined by $X$.

 Let $p \in M$ be any point. Since $\phi$ is timelike, $p \ll \phi(p)$, meaning that there exists a future-directed timelike curve $\gamma$ from $p$ to $\phi(p)$. Because $\phi$ is conformal, $\phi \circ \gamma$ is a timelike curve from $\phi(p)$ to $\phi^2(p)$. If $\phi$ is orientation reversing, then $\phi \circ \gamma$ is past-directed, and hence $\phi^2(p) \ll \phi(p)$. But by timelikeness of $\phi$, we also have $\phi(p) \ll \phi^2(p)$, in contradiction to the assumption that $(M,g)$ is chronological. Thus $\phi$ must be time orientation preserving.
\end{proof}

\begin{lem} \label{lem:xn}
 Let $(M,g)$ be a chronological spacetime with a timelike conformal transformation $\phi : M \to M$. Then, for every point $x \in M$, there exists a future-inextendible timelike curve $\gamma : [0,\infty) \to M$ with $\gamma(n) = \phi^n(x)$ for all $n \in \nat$.
\end{lem}

\begin{proof}
Let $x_n := \phi^n(x)$. Since $x_{n} \ll x_{n+1}$, there is a timelike curve $\gamma_n : [n,n+1] \to M$ from $x_n$ to $x_{n+1}$. We then construct $\gamma$ by concatenating the $\gamma_n$'s. It remains to show that $\gamma$ is future-inextendible. For this, it suffices if the sequence $(x_n)_n$ does not converge. Suppose that $x_n$ converges to a limit point $x_\infty$. We show that then $x_\infty$ is a fixed point of $\phi$, but this would imply $x_\infty \ll x_\infty$, contradicting chronology. Suppose that $\phi(x_\infty) \neq x_\infty$, then by continuity there exists a small enough neighborhood $U$ of $x_\infty$ such that $\phi(U) \cap U = \emptyset$. Convergence $x_n \to x_\infty$ implies that $x_n \in U$ for all $n$ large enough, but then, by the previous sentence, $x_{n+1} = \phi(x_n) \not\in U$, a contradiction. Therefore $x_\infty$ must be a fixed point, but there can be no fixed points, so we conclude that the sequence $x_n$ diverges, and hence $\gamma$ is inextendible.
\end{proof}

\begin{lem} \label{lem:MisTIP}
 Let $(M,g)$ be a connected spacetime which admits a timelike conformal transformation $\phi : M \to M$. Then $M$ is a TIP and a TIF of $(M,g)$.
\end{lem}

\begin{proof}
 Let $x \in M$ be a point, $x_n := \phi^n(x)$, and define the set
 \begin{equation*}
  U := \bigcup_{n=1}^\infty I^-(x_n) = \bigcup_{n=1}^\infty \phi^n \left( I^-(x) \right).
 \end{equation*}
 Clearly, $U$ is a past set, and it is invariant under $\phi$, i.e.\ $\phi(U)=U$. By Lemma~\ref{lem:xn}, there is a future-inextendible timelike curve $\gamma$ with $\gamma(n) = x_n$. Hence $U = I^-(\gamma)$ is a TIP.

 Assume that $U \neq M$. Then, since $M$ is connected, $U$ has non-empty boundary $\partial U$. We proceed to show that the existence of a $p \in \partial U$ leads to a contradiction. On the one hand, $U = \phi(U)$ implies, by continuity of $\phi$, that $\partial U = \phi(\partial U)$. On the other hand, timelikeness of $\phi$ implies $\phi(p) \in I^+(p)$. Thus $q := \phi(p) \in I^+(p) \cap \partial U$. But then, by openness of the chronological relation, there is a $q' \in I^+(p) \cap U$ (close to $q$), and $U$ being a past set then implies $p \in I^-(q') \subset U = \Int(U)$, a contradiction to the assumption that $p \in \partial U$. Hence $\partial U$ must be empty, implying that $U = M$.

 That $M$ is a TIF can be shown by the same argument, replacing $\phi$ with $\phi^{-1}$.
\end{proof}

The second part of Theorem \ref{thm:conf} states that when adding the assumption of compact Cauchy surfaces, we even get $\CB^\pm = \{M\}$. To prove this, we use that by Theorem \ref{thm:mTIP}, the spacetime $(M,g)$ contains a minimal TIP, and hence the following lemma applies.

\begin{lem}
 Let $(M,g)$ be a connected spacetime admitting a timelike conformal transformation $\phi : M \to M$. If $(M,g)$ contains a minimal TIP $U$, then $U = M$ and hence $\CB^+=\{M\}$.
\end{lem}

\begin{proof}
 Let $\gamma$ be an inextendible causal curve such that $I^-(\gamma) = U$ is the minimal TIP. Then, by timelikeness of $\phi$, we have $\phi^{-1}(U) \subseteq U$. Moreover, the curve $\phi^{-1} \circ \gamma$ is also inextendible causal, and because $\phi$ is conformal, $I^-(\phi^{-1} \circ \gamma) = \phi^{-1}(I^-(\gamma)) = \phi^{-1}(U)$. Thus $\phi^{-1}(U) \subseteq U$ is a TIP, and by minimality of $U$, we must have $\phi^{-1}(U) = U$. In other words, $U$ is invariant under $\phi$, and by the second paragraph in the proof of Lemma \ref{lem:MisTIP}, this implies $U = M$. We conclude that $M$ is a minimal TIP, and hence the causal boundary cannot contain any elements other than $M$.
\end{proof}

That $\CB^-=\{M\}$ in Theorem \ref{thm:conf} follows from the above argument, replacing TIPs with TIFs and $\phi^{-1}$ with $\phi$. Finally, that $(M,g)$ contains a line follows from Proposition \ref{prop:impliesminmin} and Theorem \ref{thm:line} (in fact, a simplified version of the proof of Theorem \ref{thm:line} is already enough, see also \cite{Tip}).

Having proven Theorem \ref{thm:conf}, we end this section with three examples. The first two are of timelike conformal transformations that are not the flow of a conformal Killing vector field. The third one is of a spacetime with non-spacelike causal boundary which admits a complete conformal Killing vector field that is timelike everywhere, except on a codimension $1$ subset, where it vanishes.

\begin{exam}
 Consider $M := \real \times S^1 \times S^1$ with metric tensor \[ g:=-dt^2+(\sin^2(t)+1) d\theta^2 + d\varphi^2.\] Then the map $\phi : M \to M$ given by $\phi(t,\theta,\varphi) := (t+2\pi, \theta, \phi)$ is a timelike conformal transformation. We proceed to show that it does not arise as the flow of a conformal Killing vector field.

 Suppose that $X = X^0 \partial_t + X^1 \partial_\theta + X^2 \partial_\varphi$ is a conformal Killing vector field, i.e.\ $\mathcal{L}_X g = \Lambda g$. From the rotation and reflection symmetry of $g$ in the two $S^1$ factors, it follows that then also $Y = X^0 \partial_t - X^1 \partial_\theta - X^2 \partial_\varphi$ must be conformal Killing, and hence also $Z := \frac{1}{2}\left(X+Y\right) = X^0 \partial_t$. But evaluating the conformal Killing equation in components, we obtain
 \begin{align*}
  \Lambda = \Lambda g_{22} = \left[ \mathcal{L}_Z g \right]_{22} = 0 &\implies \Lambda = 0, \\
  0 = \Lambda g_{01} = \left[ \mathcal{L}_Z g \right]_{01} = -\partial_\theta X^0 &\implies X^0 = X^0(t,\varphi), \\
  0 = \Lambda g_{02} = \left[ \mathcal{L}_Z g \right]_{02} = -\partial_\varphi X^0 &\implies X^0 = X^0(t), \\
  \Lambda (\sin^2(t)+1) = \Lambda g_{11} = \left[ \mathcal{L}_Z g \right]_{11} = 2 \sin(t) \cos(t) X^0 &\implies X^0 = 0.
 \end{align*}
It follows that $X$, which was an arbitrary conformal Killing vector field, cannot have the map $\phi$ as its flow.
\end{exam}

In the previous example, we have shown the absence of timelike conformal Killing fields by direct computation. In the next example, on the other hand, we use a causal theoretic argument based on a result of Javaloyes and S\'anchez \cite{JaSa}. The example below is also an example of a spacetime where the first part of Theorem \ref{thm:conf} applies ($M$ is a TIP), but not the second part ($M$ is not globally hyperbolic, and $\CB^+$ has more than one point).

\begin{exam}
 Consider a cylinder with vertical slits (i.e.\ removed line segments),
 \begin{equation*}
  M := \left(\real \times S^1 \right) \setminus \left\{ (t,0) \mid t \in [2n,2n+1] \text{ for some } n \in \nat \right\},
 \end{equation*}
 equipped with the Minkowski metric $g := -dt^2+d\theta^2$. Then $(M,g)$ is stably causal (because $t$ is a time function), and hence distinguishing. According to \cite[Thm.~1.2]{JaSa}, if $(M,g)$ were to admit a timelike conformal Killing vector field, then $(M,g)$ would automatically be causally continuous, which it clearly is not (the reflectivity property fails because of the slits). Nonetheless, $(M,g)$ does admit a timelike conformal transformation, namely $\phi(t,\theta) := (t+2,\theta)$. Note also that $(M,g)$ does admit \emph{incomplete} timelike conformal Killing vector fields, such as $\partial_t$. Finally, while $M$ is a TIP of $(M,g)$ by Theorem \ref{thm:conf}, there are also TIPs arising as boundary points of the slits.
\end{exam}

In the two previous examples, the timelike conformal transformation was given by translation of a natural time coordinate. Next we give an example that cannot be written in this form because of a topological obstruction.

\begin{exam}
 Let $n$ be odd and consider $M := \real \times S^n$ equipped with the Lorentzian metric $g := -dt^2 + h$, where $h$ is the round metric on the $n$-sphere $S^n$. Then $(M,g)$ is globally hyperbolic, and $\Sigma_0 := \{t=0\}$, $\Sigma_7 := \{t=7\}$ are compact Cauchy surfaces. Moreover, let $F \colon S^n \to S^n$ denote the antipodal map, and define
 \begin{align*}
  \phi \colon \real \times S^n &\longrightarrow \real \times S^n, \\ (t,x) &\longmapsto (t+7,F(x)),
 \end{align*}
 which is an isometry of $(M,g)$. Moreover, because the diameter of $S^n$ is less than $7$, it is easy to see that $(t+7,y) \in I^+((t,x))$ for all $t \in \real$ and all $x,y \in S^n$. Hence $\phi$ is timelike.

 Consider on $M$ any Cauchy temporal function $\tau$ with its associated splitting
 \begin{align*}
  \psi \colon M &\longrightarrow \real \times \Sigma \\ p &\longmapsto (\tau(p),\pi(p)).
 \end{align*}
 Here $\Sigma := \tau^{-1}(0)$, and $\pi(p)$ is the unique intersection point $\pi(p) \in \gamma \cap \Sigma$, where $\gamma$ is the integral curve of the gradient vector field $\nabla \tau$ through $p$. Suppose that we can choose $\tau$ such that
 \begin{equation*}
  \chi := \psi \circ \phi \circ \psi^{-1} \colon \real \times \Sigma \to \real \times \Sigma
 \end{equation*}
 takes the form
 \begin{equation*}
  \chi (\tilde t, \tilde x) = (\tilde t + T, \tilde x)
 \end{equation*}
 for some constant $T \in \real$. Clearly, $\chi$ is homotopic to the identity, but this would imply that $\phi = \psi^{-1} \circ \chi \circ \psi$ is homotopic to the identity as well, via a homotopy $H$. Then, if $\Pi : M \to S^n$ denotes the canonical projection of $M = \real \times S^n$ onto $S^n$, $\Pi \circ H$ would be a homotopy between the antipodal map $F \colon S^n \to S^n$ and the identity. This is a contradiction, because in odd dimension, the antipodal map is not homotopic to the identity. Therefore, in this example we cannot find a splitting of the spacetime where the action of $\phi$ reduces to discrete time translations.
\end{exam}

Lastly, we show that timelikeness of cannot be replaced by causality, even when talking about conformal Killing vector fields.

\begin{exam} \label{ex:flow}
 Consider the spacetime of Example \ref{ex:notminmin} (see also Figure \ref{fig:ex}). We can describe the ambient Minkowski spacetime in (rescaled) double null coordinates $p :=\frac{\pi}{2}(t-x)$, $q := \frac{\pi}{2}(t+x)$, so that our spacetime $M$ corresponds to the points $(p,q)$ with
 \begin{equation*}
  p,q < \frac{\pi}{2} \ \text{ and } \ \begin{cases}
                                    -\frac{\pi}{2} < p, -\pi+p < q &\text{or} \\
                                    -\frac{\pi}{2} < q, -\pi+q < p.
                                   \end{cases}
 \end{equation*}
In these coordinates, the metric is given by $g = -\frac{4}{\pi^2}dpdq$. The region $V$ now corresponds to $-\frac{\pi}{2} < p,q <\frac{\pi}{2}$, and is conformally equivalent to $(1+1)$-dimensional Minkowski spacetime through the standard conformal embedding
 \begin{align*}
  F : \real^{1,1} &\longrightarrow V \\
  (u,v) &\longmapsto (\arctan(u),\arctan(v))
 \end{align*}
commonly used to construct the Penrose diagramm of Minkowski spacetime. Here $(u,v)$ are again the double null coordinates, in which the Minkowski metric takes the form $\eta = -du dv$. Then $F^* \eta = \frac{-1}{\cos^2 (p) \cos^2 (q)} dp dq$, so indeed $F$ is conformal. Similarly, we can conformally embed another copy of $\real^{1,1}$ into the region $U$.

The flow along the coordinate (Killing) vector field $2 \partial_t$ in $\real^{1,1}$, pushed forward through the conformal embedding $F$, defines a conformal Killing vector field
\begin{align*}
 X := F^*(2 \partial_t) = F^*\left(\partial_u + \partial_v \right) = \frac{1}{1+\tan^2(p)} \partial_p + \frac{1}{1+\tan^2(q)} \partial_q
\end{align*}
on $V$. This vector field $X$ can be extended to the boundary of $V$, and the extension is complete (and follows the lightlike direction along the boundary). Repeating the same construction on $U$, we construct a vector field, denoted again by $X$, on the whole spacetime $(M,g)$. The flowlines of $X$ are depicted in Figure \ref{fig:flow}. That $X$ is conformal Killing is more or less obvious by construction, at least in the interior of $U$ and $V$. However, one should be a bit careful on the boundary between the two, where we cannot interpret $X$ as the pushforward of $2 \partial_t$. Nonetheless, a simple computation shows that
\begin{align*}
 \left[\mathcal{L}_X g\right]_{11} = \left[\mathcal{L}_X g\right]_{22} = 0, && \left[\mathcal{L}_X g\right]_{12} = \left[\mathcal{L}_X g\right]_{21},
\end{align*}
and hence $X$ is conformal Killing for the metric $g = -dpdq$.

We have thus shown that our spacetime $(M,g)$, despite having non-spacelike causal boundary, admits a complete causal conformal Killing vector field. The key fact that makes the proof of Theorem \ref{thm:conf} fail in this example, is that the minimal TIP $U$ is invariant under the flow of $X$.
\end{exam}

\section*{Declaration}

This version of the article has been accepted for publication after peer review, but is not the Version of Record and does not reflect post-acceptance improvements, or any corrections. The Version of Record is available online at: \href{http://dx.doi.org/10.1007/s00009-025-02832-3}{http://dx.doi.org/10.1007/s00009-025-02832-3}

\bibliographystyle{abbrv}
\bibliography{refs.bib}

\end{document}